\documentclass[pra,amsmath,amsfonts,onecolumn,superscriptaddress,showpacs, showkeys]{revtex4-1}
\usepackage[psamsfonts]{amssymb}
\usepackage{amsmath}
\usepackage{graphicx}
\usepackage{color,soul}
\usepackage{xcolor}

\newtheorem{theorem}{Theorem}
\newtheorem{lemma}[theorem]{Lemma}
\newtheorem{proposition}[theorem]{Proposition}

\newenvironment{proof}[1][Proof]{\begin{trivlist}
\item[\hskip \labelsep {\bfseries #1}]}{\end{trivlist}}

 \newcommand{\ket}[1]{|#1\rangle}
 \newcommand{\bra}[1]{\langle #1|}
 
 \newcommand{\braket}[2]{\langle#1|#2\rangle}

 \newcommand{\Id}{{\mathbb I}}

 \newcommand{\T}{{\mathrm t}}

\newcommand{\Tr}{{\mathrm {Tr}}}

\newcommand{\etal}{\textit {et al.} }

\begin{document}
\title{Quantifying Nonclassicality of Correlations  based on the Concept of Nondisruptive Local State Identification}
\author{Azam Kheirollahi\footnote{a.kheirollahi@sci.ui.ac.ir}}\affiliation{Department of Physics, University of Isfahan,
 Isfahan, Iran}
\author{Seyed Javad Akhtarshenas\footnote{akhtarshenas@um.ac.ir}}
\affiliation{Department of Physics, Ferdowsi University of Mashhad,
 Mashhad, Iran}
 \author{Hamidreza Mohammadi\footnote{hr.mohammadi@sci.ui.ac.ir}}\affiliation{Department of Physics, University of Isfahan,
 Isfahan, Iran}
\affiliation{Quantum Optics Group, University of Isfahan,
 Isfahan, Iran}

\begin{abstract}
A bipartite state is classical with respect to party $A$ if and only if party $A$ can perform  nondisruptive local state identification (NDLID) by a projective measurement. Motivated by this we  introduce a  class of quantum correlation measures for an arbitrary bipartite state. The measures utilize   the general Schatten $p$-norm to quantify the amount of departure from the  necessary and sufficient condition of classicality of correlations provided by the concept of NDLID. We show that for the case of Hilbert-Schmidt norm,  i.e. $p=2$, a closed formula is available  for an arbitrary bipartite state. The reliability of the proposed measures is checked from the information theoretic perspective. Also, the monotonicity behavior of these measures under LOCC is exemplified.   The results reveal that for the general pure bipartite states these measures have an upper bound which is an entanglement monotone in its own right. This enables  us to introduce  a new measure of entanglement, for a general bipartite  state, by  convex roof construction. Some examples and comparison with other quantum correlation measures  are also provided.
\end{abstract}


\maketitle

\section{Introduction}
The most significant feature of quantum systems is the quantum superposition. This property of quantum mechanics  arises from the linearity of quantum mechanics and is the origin of the quantum correlation  in composite quantum systems. For decades, the notion of quantum correlation was often associated with the concept of entanglement.  Entanglement is an important resource for quantum information and computation processing and is necessary for performance of some quantum communication protocols \cite{Barnett-book}. However, entanglement is not the only aspect of quantum correlations; some separable (disentangled) states exhibit nonclassical features \cite{Zurek2001, Henderson2001}.
A great deal of works has been spent to the subject of the measures of correlations (see \cite{Modi2012, Celeri2011} and references therein), and various measures of quantum correlations beyond  entanglement have been introduced, some of them are known under the collective name quantum discord. Many of these measures have been related to various tasks and concepts in quantum information and quantum computation such as, decoherence \cite{Zurek2001}, measurement induced non locality \cite{Luo2011}, geometry of state space \cite{Dakic2010, Akhtarshenas2015}, state discrimination \cite{Li2012},  deterministic quantum computation with one qubit \cite{Knill1998, Datta2008, Fanchini2011}, witnessing the quantum correlation \cite{Fazio2013}, no-broadcasting \cite{Piani2008,Piani2009, Luo2010}, quantum metrology \cite{Modi2011}, quantum state merging \cite{Horodecki2005, Cavalcanti2011, Madhok2011}, and quantum thermodynamics \cite{Zurek2003, Brodutch2010}. There have also been several proposals related to experimental investigations of classical correlation and quantum correlation beyond entanglement \cite {Lanyon2008, Xu2010, Rahimi2010, Yu2011, Auccaise2011, Auccaise2011-1, Dakic2012}.

The space of classically correlated states is a measure-zero subspace of the space of separable states \cite{FerraroPRA2010}. The state $\rho$ of a bipartite system is called classical-quantum if it is classical only  with respect to the party $A$, i.e. if and only if it can be represented as $\rho=\sum_ip_i\Pi_i^{A}\otimes \rho_i^{B}$,
with $\Pi_i^A=\ket{i}\bra{i}$ as the projection operator on the orthonormal basis of $\mathcal{H}^{A}$, and $\rho_i^B$ being a state on $\mathcal{H}^{B}$.
The same definition holds for the quantum-classical states, i.e. states that are classical only with respect to the party $B$.
A state  is classically correlated  if and only if it is both classical-quantum  and quantum-classical state.

In order to distinguish classically correlated states from the set of quantum states,  Chen \etal \cite{ChenPRA2011} have introduced the concept of nondisruptive local state identification (NDLID). A bipartite state $\rho$  is classical with respect to party $A$ if and only if party $A$ can perform  NDLID by a projective measurement \cite{ChenPRA2011}.  They showed that the states which can undergo NDLID task are locally broadcastable (see \cite{Piani2008} for local broadcasting) and hence are classical states, i.e. they are classically correlated states. Accordingly,  they  provided the following theorem in order to decide whether or not a given bipartite state $\rho$, acting on the Hilbert space $\mathcal{H}=\mathcal{H}^A\otimes \mathcal{H}^B$,  is classical with respect to the party $A$.
\begin{theorem}\label{TheoremChenPRA} \cite{ChenPRA2011}
Let $\rho$ be a bipartite state acting on the Hilbert space $\mathcal{H}=\mathcal{H}^A\otimes \mathcal{H}^B$. Let also $\Phi^{(B)}=\{\ket{\phi_{i}^{(B)}}\}_{i=1}^{d_B}$ denotes  any orthonormal basis for $\mathcal{H}^{B}$. Then $\rho$ is classical with respect to party $A$ if and only if
 \begin{equation}\label{Aij}
A_{ij}^{\Phi^{(B)}}(\rho):=\bra{\phi_{i}^{(B)}}\rho\ket{\phi_{j}^{(B)}},
\end{equation}
is diagonal in the same orthonormal basis  $\{\ket{a_{k}}\}_{k=1}^{d_A}$ for all $i,j$.
\end{theorem}
To characterize classically correlated states, Wu \etal  have obtained similar results in \cite{Wu2011}  and proposed a norm-based measurement of quantum correlation of two-qubit states. They have used max norm of operators to quantify the deviation from the necessary and sufficient condition for classical correlated state, and investigated the dynamics of quantum correlations in Markovian and non-Markovian processes.  In a similar manner,  Guo \etal \cite{Guo2012}  introduced a different measure of quantum correlation  by using Hilbert-Schmidt norm.
While both of  these measures are computable, they did not include an in-depth analysis of the correlation measures from an information theoretic perspective; they are not judged according to any information theoretic criteria like the criteria provided in Ref. \cite{Brodutch2012}. More precisely, the former is base-dependent, so  it is not invariant under local unitary transformation and  can not be considered as a reliable measure.  The latter, however,  does not reduce to an entanglement monotone for pure states.

In this paper we use the concept of NDLID and provide a class of quantifiers of quantum correlation for an arbitrary bipartite state.
As a bipartite state $\rho$ can undergo NDLID by party $A$  if and only if it is classical with respect to party $A$ \cite{ChenPRA2011},  any disability of such task comes from the nonclassical correlation of the party $A$. Exploiting this notion, we define a measure of quantum correlation by quantifying   the amount that the state violates the  necessary and sufficient condition of classicality of correlation, stated in the theorem above. We utilize a general Schatten $p$-norm \cite{Schatten-norm} to quantify the degree of non-commutativity of the operators $A_{ij}^{\Phi^{(B)}}(\rho)$ in Eq. \eqref{Aij}. Followed by the minimization over the orthonormal basis $\Phi^{(B)}$, we lead to a  class of quantum correlation quantifiers which are invariant under local unitary transformations. We show that for the Hilbert-Schmidt norm, i.e. $p=2$, the defined measure does not require optimization, leading to a closed relation in this case.
We also  show that for an arbitrary $p$, our measures are  non-increasing upon attaching local ancillary state on the unmeasured subsystem.    Furthermore, we find that the one-norm, i.e. $p=1$, is the case that the measure remains invariant upon attaching to or removing of local ancillary state on the unmeasured subsystem. The monotonicity behavior of the  measures   under local operations and classical communications (LOCC) is also exemplified and we find that for two-qubit case the measures are monotone for $p\leq3$. In addition we find that, at least for $d\le 3$,  the measures are monotone  for $p=1$. This result enables us to define an entanglement measure by convex roof construction.

The remainder of this paper is organized as follows. In section II, we introduce our measures of quantum correlation and provide a closed relation for the case of Hilbert-Schmidt norm. Section III is devoted to investigate some properties of the measures. Some examples are given in section IV. The paper is concluded in section V.

\section{ Quantifying quantum correlation through the Schatten $p$-norm}
  Theorem \ref{TheoremChenPRA}  provides a necessary and sufficient condition for classicality of a bipartite state $\rho$ due to party $A$, in the sense that
 $\rho$ is classic with respect to party $A$ if and only if  $A_{ij}^{\Phi^{(B)}}(\rho)$ is diagonal in the same orthonormal basis $\{\ket{a_{k}}\}_{k=1}^{d_A}$ for all $i,j$.
 But the set   $\mathcal{A}^{\Phi^{(B)}}(\rho)=\{A_{ij}^{\Phi^{(B)}}(\rho)\}_{i,j=1}^{d_B}$ of operators have a simultaneous eigenvectors if and only if they are all normal operators, i.e. $[A_{ij}^{\Phi^{(B)}}(\rho),{A_{ij}^{\Phi^{(B)}}}^\dagger(\rho)]=0$ for $i,j=1,\cdots,d_B$, and that all  operators commute by pairs, i.e. $[A_{ij}^{\Phi^{(B)}}(\rho),A_{kl}^{\Phi^{(B)}}(\rho)]=0$ for all pairs $ij$ and $kl$. However, since ${A_{ij}^{\Phi^{(B)}}}^\dagger(\rho)=A_{ji}^{\Phi^{(B)}}(\rho)$, i.e. the above set is closed under Hermitian adjoint,  so that this theorem implies that $\rho$ is classical with respect to party $A$ if and only if for any orthonormal basis $\Phi^{(B)}$ of $\mathcal{H}^B$ the commutator    $[A_{ij}^{\Phi^{(B)}}(\rho),A_{kl}^{\Phi^{(B)}}(\rho)]$ vanishes  for all pairs $ij$ and $kl$. In other words, $\rho$ is classical with respect to party $A$ if and only if for any orthonormal basis $\Phi^{(B)}$ of $\mathcal{H}^B$ the set  $\mathcal{A}^{\Phi^{(B)}}(\rho)$ forms a set of commuting operators.

Measuring any departure from  this condition may be used as an indicator of the quantumness of the system.
In order to quantify any violation of this condition, we use the general Schatten $p$-norm and quantify the order of non-commutativity of the set $\mathcal{A}^{\Phi^{(B)}}(\rho)$. Let $\mathcal{M}^{\Phi^{(B)}}(\rho)=\{M_{ij,kl}^{\Phi^{(B)}}(\rho)\}$, with  $M_{ij,kl}^{\Phi^{(B)}}(\rho)=[A_{ij}^{\Phi^{(B)}}(\rho),A_{kl}^{\Phi^{(B)}}(\rho)]$, denotes a set of operators obtained from the pairwise commutators of all entities  of $\mathcal{A}^{\Phi^{(B)}}(\rho)$.  Using the collective index $I=\{ij,kl\}$ for entities of $\mathcal{M}^{\Phi^{(B)}}(\rho)$, we write  the Schatten $p$-norm of $M_{I}^{\Phi^{(B)}}(\rho)$ as \cite{Schatten-norm}
\begin{equation}\label{DpRhoPhiBI}
D_{p}[M_{I}^{\Phi^{(B)}}(\rho)]:=\left\|M_{I}^{\Phi^{(B)}}(\rho)\right\|_p=\left[\Tr\left(M_{I}^{\Phi^{(B)}}(\rho)
{M_{I}^{\Phi^{(B)}}}^\dagger(\rho) \right)^{\frac{p}{2}}\right]^{\frac{1}{p}}.
\end{equation}
Theorem \ref{TheoremChenPRA} then implies that $\rho$ is a classical state with respect to party $A$ if and only if the above quantity vanishes for all entities of the set $\mathcal{M}^{\Phi^{(B)}}(\rho)$, i.e. for any pair of indices $I=\{ij,kl\}$. Accordingly, we define
\begin{equation}\label{DpRhoPhiB}
D_{p}^{\Phi^{(B)}}(\rho):=\left[\sum_{I}^{\;\;\quad\prime} \left(D_{p}[M_{I}^{\Phi^{(B)}}(\rho)]\right)^p\right]^{1/p}=\left[\sum_{I}^{\;\;\quad\prime}\Tr\left(M_{I}^{\Phi^{(B)}}(\rho){M_{I}^{\Phi^{(B)}}}^{\dagger}(\rho) \right)^{\frac{p}{2}}\right]^{1/p},
\end{equation}
as an indicator of the quantumness of the correlation of $\rho$. Here, we used $\sum_{I}^{\prime}$ to stress that the sum is performed over all inequivalent nontrivial pairs of $I=\{ij,kl\}$, i.e. for all pairs  such that $\{ij,kl\}\ne\{kl,ij\}$, in order to avoid double counting. Evidently, $\rho$ is a classical-quantum state if and only if $D_{p}^{\Phi^{(B)}}(\rho)=0$ for any orthonormal basis $\Phi^{(B)}$ of party $B$. However, the above quantity depends on the chosen basis, so that to make it independent on the basis of the party $B$, we propose the following quantity as a measure of the quantumness of the correlation.

\begin{proposition}
For any bipartite state $\rho$ we define
\begin{equation}\label{DpRho}
D_{p}(\rho)=\min _{\Phi^{(B)}}D_{p}^{\Phi^{(B)}}(\rho),
\end{equation}
as a measure of the quantumness of the correlation of $\rho$ with respect to the party $A$. Here the minimum is taken over  any orthonormal basis for $\mathcal{H}^{B}$.
\end{proposition}
Before discussing various properties of $D_{p}(\rho)$, let us mention that in the particular case $p=2$ the definition \eqref{DpRho} does not require  minimization, i.e.  $D_{p=2}(\rho)=D_{p=2}^{\Phi^{(B)}}(\rho)$ for any basis ${\Phi^{(B)}}$. In Appendix \ref{AppendixProofD=2} we will provide a proof for this assertion, along  with  a closed relation  for $D_{p=2}(\rho)$. The result is summarized in the following theorem.
\begin{theorem}\label{D2FATheorem}
For an arbitrary bipartite state $\rho$, we find the following closed relation for $D_{p=2}(\rho)$
\begin{eqnarray}\label{sqD2FA}
D_{p=2}(\rho)=\frac{2}{d_A^2 d_B^2}\sqrt{
-\Tr{\left\{\mathcal{F}^A(\rho)\left[d_B \vec{x}\vec{x}^\T+TT^\T\right]\right\}}},
\end{eqnarray}
where $\vec{x}$ is the local coherence vector of party $A$, and $T$ denotes correlation matrix of the state $\rho$,  defined by Eqs. \eqref{xy} and \eqref{T}, respectively. Moreover,  $\mathcal{F}^A(\rho)=\sum_{r=1}^{d_A^2-1}\left[F^A_r(TT^\T){F^A_r}^\dagger\right]$ where
$(F^A_r)_{pq}=-if^A_{pqr}$ with $f^A_{pqr}$ as the structure constant of the Lie algebra $SU(d_A)$ (see \eqref{SUmGellMann}).
\end{theorem}

\section{Properties of $D_p(\rho)$}
In this section we investigate  some properties of $D_{p}(\rho)$. To make these properties clearer, we first discuss  the properties of $D_{p}(\rho)$ for a general pure state $\rho=\ket{\psi}\bra{\psi}$.  Then we will discuss the properties of $D_{p}(\rho)$ for a general mixed state $\rho$.
\subsection{Properties of $D_p(\rho)$: Pure states}
 Let $\ket{\psi}=\sum_{m=1}^{d}\sqrt{\lambda_m}\ket{e_m^A} \ket{e_m^B}$ be a general pure state in its Schmidt representation. We find $A^{\Phi^{(B)}}_{ij}(\psi)=\ket{\xi_i^A}\bra{\xi_j^A}$ where $\ket{\xi_i^A}=\braket{\phi_i^B}{\psi}=\sum_{m=1}^d\alpha_m^{(i)}\ket{e_m^A}$ with $\alpha_m^{(i)}=\sqrt{\lambda_m}\braket{\phi_i^B}{e_m^B}$. Interestingly, the set of vectors $\{\ket{\xi_i^A}\}_{i=1}^{d}$ gives us both reduced density matrices $\rho^A$ and $\rho^B$ as  $\rho^A=\sum_{i=1}^{d}\ket{\xi_i^A}\bra{\xi_i^A}$ and $\rho^B_{\Phi^{(B)}_{ij}}=\bra{\phi_i^B}\rho^B\ket{\phi_j^B}=\braket{\xi_j^A}{\xi_i^A}$, respectively. Using these definitions we get, for a fixed $I=\{ij,kl\}$
 \begin{eqnarray}\label{MMdXiA}\nonumber
 M_{I}^{\Phi^{(B)}}(\psi){M_{I}^{\Phi^{(B)}}}^\dagger(\psi)&=&\rho^B_{\Phi^{(B)}_{jk}}\rho^B_{\Phi^{(B)}_{kj}}\rho^B_{\Phi^{(B)}_{ll}}\ket{\xi_i^A}\bra{\xi_i^A}
 -\rho^B_{\Phi^{(B)}_{jk}}\rho^B_{\Phi^{(B)}_{il}}\rho^B_{\Phi^{(B)}_{lj}}\ket{\xi_i^A}\bra{\xi_k^A} \\
 &-& \rho^B_{\Phi^{(B)}_{li}}\rho^B_{\Phi^{(B)}_{kj}}\rho^B_{\Phi^{(B)}_{jl}}\ket{\xi_k^A}\bra{\xi_i^A}+
 \rho^B_{\Phi^{(B)}_{li}}\rho^B_{\Phi^{(B)}_{il}}\rho^B_{\Phi^{(B)}_{jj}}\ket{\xi_k^A}\bra{\xi_k^A},
 \end{eqnarray}
 where can be used in Eq. \eqref{DpRhoPhiB} to obtain $D_{p}^{\Phi^{(B)}}(\psi)$ for an arbitrary  bipartite  pure state $\ket{\psi}$, and in any basis  $\Phi^{(B)}$ of $\mathcal{H}^B$.
 Using the above relation we provide a tight  upper bound for $D_{p}(\psi)$.
 \begin{lemma}\label{DpPureLemma}
For a general pure state $\ket{\psi}$ with Schmidt numbers $\{\lambda_m\}_{m=1}^{d}$, the quantum correlation $D_{p}(\psi)$ is bounded from above as
\begin{equation}\label{DpPureUB}
D_{p}(\psi)\le \left[\sum_{i<k}(\lambda_i\lambda_k)^{p/2}\left((\lambda_i^p+\lambda_k^p)+\mu_{p/2}(\psi)(\lambda_i^{p/2}+\lambda_k^{p/2})\right)\right]^{1/p},
\end{equation}
where we have defined $\mu_{q}(\psi)=\Tr{(\rho^B)^{q}}=\sum_{m=1}^d\lambda_{m}^{q}$.
\end{lemma}
\begin{proof}
For an arbitrary $\ket{\psi}$ let us  choose the basis $\Phi^{(B)}$  as the local Schmidt basis (LSB) of $\rho^B$, i.e. $\Phi^{(B)}=\{\ket{e_m^B}\}_{m=1}^{d}$ so that $\braket{\phi_i^B}{e_m^B}=\delta_{im}$.
In this case we have $\ket{\xi_m^A}=\sqrt{\lambda_m}\ket{e_m^A}$ for $m=1,\cdots,d$. Using this and Eqs. \eqref{DpRhoPhiB} and \eqref{MMdXiA} we get
\begin{equation}\label{DpPureLSB}
D_{p}^{\textrm{LSB}}(\psi)=\left[\sum_{i<k}(\lambda_i\lambda_k)^{p/2}\left((\lambda_i^p+\lambda_k^p)+\mu_{p/2}(\psi)(\lambda_i^{p/2}+\lambda_k^{p/2})\right)\right]^{1/p},
\end{equation}
where by Eq. \eqref{DpRho} leads to Eq. \eqref{DpPureUB}.
\end{proof}
Note that the above upper bound is tight in the sense that there exist states for which the bound is saturated. In particular, one can easily shows that the bound is tight for the following cases.
\begin{enumerate}
\item
 For arbitrary values of $p$ and $d$,  the bound reduces to  zero for  the product state $\ket{\psi_{\textrm{pro}}}=\ket{e^A}\ket{e^B}$.
\item
For arbitrary values of  $p$ and $d$, the bound is saturated   for  the  maximally entangled state $\ket{\psi_{\max}}=\frac{1}{\sqrt{d}}\sum_{m=1}^{d}\ket{e_m^A}\ket{e_m^B}$ as
\begin{equation}\label{DpPureMax}
{D_{p}}(\psi_{\max})=\frac{1}{d^2}[d(d^2-1)]^{\frac{1}{p}}.
\end{equation}
To see this recall  that in this case we have $\rho^B_{\Phi^{(B)}_{ij}}=\frac{1}{d}\delta_{ij}=\braket{\xi_i^A}{\xi_j^A}$, irrespective of the chosen basis $\Phi^{(B)}$. This, however, can be used to define  an orthonormal basis for $\mathcal{H}^A$ as $\{\ket{\hat{\xi}_i^A}=\sqrt{d}\ket{\xi_i^A}\}_{i=1}^{d}$. Using this and Eqs.  \eqref{DpRhoPhiB}, \eqref{MMdXiA}, and after some straightforward calculations,  one can find Eq. \eqref{DpPureMax} which, clearly, coincides with the upper bound \eqref{DpPureUB}.
\item
Interestingly, when $p=1$, the bound is also tight for a general two-qubit pure state $\ket{\psi}=\sqrt{\lambda}\ket{00}+\sqrt{1-\lambda}\ket{11}$ as
\begin{equation}\label{OneNormPure}
D_{p=1}(\psi)=D_{p=1}^{\textrm{LSB}}(\psi)=2\sqrt{\lambda(1-\lambda)}\left(1+\sqrt{\lambda(1-\lambda)}\right).
\end{equation}
To see this let us choose  $\ket{\phi_1^B}=\cos \theta \ket{0}+e^{i\phi}\sin \theta \ket{1}$ and
$\ket{\phi_2^B}=\sin \theta \ket{0}-e^{i\phi}\cos \theta \ket{1}$ as a general basis for $\mathcal{H}^B$. It turns out that
\begin{eqnarray}
D_{p=1}^{\Phi^{(B)}}(\psi)=\quad\quad\quad\quad\quad\quad\quad\quad\quad\quad\quad\quad\quad\quad\quad\quad\quad\quad
\quad\quad\quad\quad\quad\quad\quad\quad\quad\quad\quad\quad\quad\\\nonumber
\sqrt{\lambda(1-\lambda)}\left\{\left[2+\vert(1-2\lambda)\sin{2\theta}\vert\right]
+\frac{1}{\sqrt{2}}\sqrt{1+4\lambda(1-\lambda)+(4\lambda(1-\lambda)-1)\cos{4\theta}}\right\},
\end{eqnarray}
depends only on the angle $\theta$. Minimum occurs for  $\theta=0$ or $\pi/2$, leads to  Eq. \eqref{OneNormPure}.
\end{enumerate}
The  following  lemma  concerns about monotinicity of these measures under LOCC operations. In particular, for two cases (i) $d=2$,  $p\le 3$ and (ii) $d=3$,  $p=1$ we have the following lemma.
\begin{lemma}\label{LemmaMonotone}
When (i) $d=2$,  $p\le 3$ and (ii) $d=3$,  $p=1$, $D_{p}^{\textrm{LSB}}(\psi)$  defines an entanglement monotone, i.e. it is a non-increasing quantity under LOCC.
\end{lemma}
\begin{proof} (i) For $d=2$  Eq. \eqref{DpPureUB} reduces to $D_{p}^{\textrm{LSB}}(\psi)=\left[\frac{1}{2}\left([\mu_{p/2}(\psi)]^4-[\mu_{p}(\psi)]^2\right)\right]^{1/p}$ which has a unique maximum  at $\lambda_1=\lambda_2=\frac{1}{2}$ only for $p\le 3$, i.e. $D_{p\le 3}^{\textrm{LSB}}(\psi)\le D_{p\le 3}^{\textrm{LSB}}(\psi_{\max})$. Moreover, recall that a function $F(\boldsymbol\lambda)$ is monotonously decreasing under LOCC if $F$ is invariant under any permutation of Schmidt coefficients $\lambda_i$ and if $F$ is Schur concave \cite{Zyczkowski2006,Buchleitner2009}, i.e.
$(\lambda_1-\lambda_2)\left(\frac{\partial F}{\partial\lambda_1}-\frac{\partial F}{\partial\lambda_2}\right)\le 0$ for all $\boldsymbol\lambda=\{\lambda_1,\cdots,\lambda_d\}$. It turns out that for $d=2$
\begin{eqnarray}
(\lambda_1-\lambda_2)\left(\frac{\partial D_{p}^{\textrm{LSB}}(\psi)}{\partial\lambda_1}-\frac{\partial D_{p}^{\textrm{LSB}}(\psi)}{\partial\lambda_2}\right)=\\\nonumber
 -(\lambda_1-\lambda_2)^{\frac{1}{2}}\left[2(\lambda_1^p+(\lambda_1\lambda_2)^{p/2}+\lambda_2^p)\right]^{\frac{1-p}{p}}
& &\left[(\lambda_1-\lambda_2)\left([\mu_{p/2}(\psi)]^2+2\mu_{p}(\psi)\right)-2(\lambda_1^{p+1}-\lambda_2^{p+1})\right],
\end{eqnarray}
which is nonpositive for all values of $\lambda_1,\lambda_2$ if $p\le 3$. This completes the proof. A similar proof can be made for the second case  $d=3$, $p=1$.
\end{proof}
The extension of the above lemma to arbitrary dimension $d$ requires more investigations.
As the degree of entanglement of any pure state may be characterized by  any Schur concave function of the Schmidt vector $\boldsymbol\lambda$ \cite{Zyczkowski2006},    one can use $D_p^{\textrm{LSB}}(\psi)$ to define an entanglement measure for a general mixed state $\rho$ by convex roof construction \cite{UhlmannOSID1998}.
\begin{theorem}
For the $p$ and $d$ expressed  by  lemma \ref{LemmaMonotone}, the upper bound $D_{p}^{\textrm{LSB}}(\psi)$, defined by Eq. \eqref{DpPureLSB}, is an entanglement monotone. Therefor we define entanglement of the bipartite pure state $\ket{\psi}$ as $E_p(\psi)=D_{p}^{\textrm{LSB}}(\psi)$. Furthermore, by convex roof construction \cite{UhlmannOSID1998} one can define a measure of entanglement of a general bipartite state $\rho$ as
\begin{equation}
E_p(\rho)=\inf\sum_{i}p_i E_p(\psi_i),\quad \sum_{i}p_i=1, \quad p_i\ge 0,
\end{equation}
where the infimum is taken over all pure state decomposition of $\rho$, i.e. all ensembles $\{p_i,\psi_i\}$ for which $\rho=\sum_i p_i\ket{\psi_i}\bra{\psi_i}$.
\end{theorem}

Let us exemplify our results by considering the two-qubit  pure state $\ket{\psi}=\sqrt{\lambda}\ket{00}+\sqrt{1-\lambda}\ket{11}$. In this case $D_{p=1}(\psi)$ is given by Eq. \eqref{OneNormPure}, and  one can find
$D_{p=2}(\psi)=\frac{1}{\sqrt{2}}\sqrt{1-[\mu(\rho^B)]^2}$.  Recall that in this case  concurrence \cite{WoottersPRL1998} is defined by $C(\psi)=2\sqrt{\lambda(1-\lambda)}=\sqrt{2(1-\mu(\rho^B))}$.
In Fig. \ref{FigTwoQubitMeasures} we have plotted $D_{p=1}(\psi)$, $D_{p=2}(\psi)$,  and  $C(\psi)$ in terms of $\sqrt{\lambda}$. Clearly, both defined measures are monotone functions of concurrence.
\begin{figure}[ht!]
\centering
\includegraphics[width=9cm]{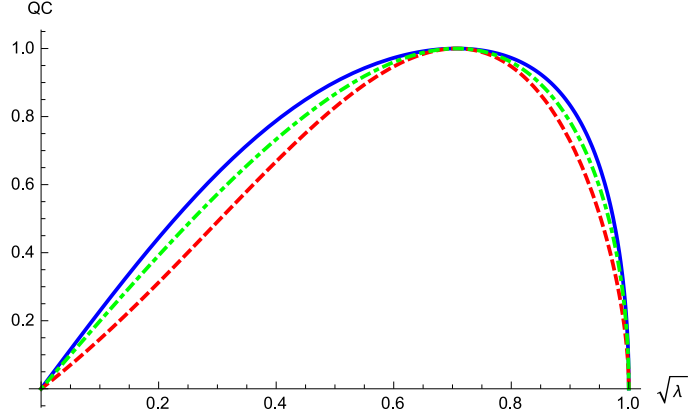}
\caption{(Color online) Quantum correlations $D_{P=2}(\psi)$ (blue, solid line), $D_{p=1}(\psi)$ (red, dashed line),  and  the concurrence $C(\psi)$ (green, doted-dashed line) of a general two-qubit pure state in terms of $\sqrt{\lambda}$. For comparison, all measures are normalized to one.}
\label{FigTwoQubitMeasures}
\end{figure}

\subsection{Properties of $D_p(\rho)$: Mixed states }
Let us now turn our attention on the  properties of the above measures for the general case of an arbitrary mixed state $\rho$. As mentioned in \cite{Brodutch2012}, a good measure of quantum correlation  should have some necessary properties. In the following we check these properties for our measure.
\begin{enumerate}
\item  Positivity, i.e. $D_{p}(\rho)\geq 0$, and the equality is satisfied if and only if the state is a classical-quantum state.

\item The measure takes its maximum value ${D_{p}}_{\max}(\rho)=\frac{1}{d^2}[d(d^2-1)]^{\frac{1}{p}}$ only for the maximally entangled
states $\ket{\psi}=\frac{1}{\sqrt{d_A}}\sum_{j=1}^{d_A}\ket{jj}$ if (i) $d=2$, $p\le 3$  or (ii) $d=3$, $p=1$. This follows easily from lemma \ref{LemmaMonotone} and the fact that for the maximally entangled states the bound is saturated.
\item
Invariance under local unitary transformations $U_{A}\otimes U_{B}$, i.e. $D_{p}(\rho)=D_{p}((U_{A}\otimes U_{B})\rho (U^{\dag}_{A}\otimes U^{\dag}_{B}))$. This follows from the fact that under such transformations $\rho\rightarrow \rho^\prime= (U_{A}\otimes V_{B})\rho(U^{\dag}_{A}\otimes V^{\dag}_{B})$, then  $A^{\Phi^{(B)}}_{ij}(\rho)\rightarrow {A^\prime}^{{\Phi^\prime}^{(B)}}_{ij}(\rho)=U_{A} A^{{\Phi^\prime}^{(B)}}_{ij}(\rho)U^{\dag}_{A}$ where ${\Phi^\prime}^{(B)}=\{U_B^\dagger\ket{\phi_{i}^{(B)}}\}$. This leads to $M^{\Phi^{(B)}}_{I}(\rho)\rightarrow {M^\prime_I}^{{\Phi^\prime}^{(B)}}(\rho)=U_A M_I^{{\Phi^\prime}^{(B)}}(\rho) U^\dagger_A$ so that
$D_{p}^{\Phi^{(B)}}(\rho)\rightarrow D_{p}^{{\Phi^\prime}^{(B)}}(\rho)$. Invoking the definition \eqref{DpRho}, we find that
$D_{p}(\rho^\prime)=\min_{\Phi^{(B)}}D_{p}^{{\Phi^\prime}^{(B)}}(\rho)=\min_{{\Phi^\prime}^{(B)}}D_{p}^{{\Phi^\prime}^{(B)}}(\rho)=D_{p}(\rho)$. This completes the assertion.

\item
No increase upon attaching a local ancillary state $\rho^C$ on the unmeasured subsystem (see \cite{Piani2012} ), i.e.  for any map
$\Gamma^C : \rho\rightarrow \rho\otimes \rho^C$ (any channel that introduces a noisy ancillary state $\rho^C$ on the unmeasured subsystem),  we have that $D_{p}(\Gamma^C(\rho))\le D_{p}(\rho)$. Moreover,  for $p=1$ the measure  is invariant under local reversible operations on the unmeasured subsystem, i.e. $D_{p=1}(\Gamma^C(\rho))=D_{p=1}(\rho)$ for any map $\Gamma^C$ that append/remove any ancillary state $\rho^C$ on/from the unmeasured subsystem $B$.
To show this let  $\Phi^{BC}=\{\ket{\phi_i^{(B)}}\ket{\phi_{i^\prime}^{(C)}}\}$ be an orthonormal basis of $\mathcal{H}^B\otimes \mathcal{H}^C$, we find
$A_{ii^\prime ,j j^\prime}^{\Phi^{(BC)}}(\Gamma^C(\rho))=\bra{\phi_{i}^{(B)}}\bra{\phi_{i^\prime}^{(C)}}\rho\otimes\rho^C\ket{\phi_{j}^{(B)}}\ket{\phi_{j^\prime}^{(C)}}$
=$A_{ij}^{\Phi^{(B)}}(\rho)\rho_{i^\prime j^\prime}^C$ where $A_{ij}^{\Phi^{(B)}}(\rho)$ is defined by Eq. \eqref{Aij} and
$\rho_{i^\prime j^\prime}^C=\bra{\phi_{i^\prime}^{(C)}}\rho^C\ket{\phi_{j^\prime}^{(C)}}$. Therefore,
$M_{II^\prime}^{\Phi^{(BC)}}(\Gamma^C(\rho))=[A_{ii^\prime ,j j^\prime}^{\Phi^{(BC)}}(\Gamma^C(\rho)),A_{kk^\prime ,l l^\prime}^{\Phi^{(BC)}}(\Gamma^C(\rho))]
=M_{I}^{\Phi^{(B)}}(\rho)\rho_{i^\prime j^\prime}^C\rho_{k^\prime l^\prime}^C$, where we have defined $I=\{ij,kl\}$, $I^\prime=\{i^\prime j^\prime,k^\prime l^\prime\}$ for the sake of simplicity. Using Eq. \eqref{DpRhoPhiBI}, this leads to
\begin{eqnarray}
D_{p}^{\Phi^{(BC)}}[M_{II^\prime}^{\Phi^{(BC)}}(\Gamma^C(\rho))]
&=&\left[\Tr\left(M_{II^\prime}^{\Phi^{(BC)}}(\Gamma^C(\rho)){M_{II^\prime}^{\Phi^{(BC)}}}^\dagger(\Gamma^C(\rho)) \right)^{\frac{p}{2}}\right]^{\frac{1}{p}} \\\nonumber
&=&\left[\Tr\left(M_{I}^{\Phi^{(B)}}(\Gamma^C(\rho)){M_{I}^{\Phi^{(B)}}}^\dagger(\Gamma^C(\rho)) \right)^{\frac{p}{2}}\left(\rho^C_{i^\prime j^\prime}\rho^C_{j^\prime i^\prime}\rho^C_{k^\prime l^\prime}\rho^C_{l^\prime k^\prime}\right)^{\frac{p}{2}}\right]^{\frac{1}{p}},
\end{eqnarray}
which can be used to write
\begin{eqnarray}\nonumber
D_{p}^{\Phi^{(BC)}}(\rho\otimes \rho^C)&=&\left[\frac{1}{2}\sum_{I}\sum_{I^\prime}\left(D_{p}^{\Phi^{(BC)}}[M_{II^\prime}^{\Phi^{(BC)}}(\Gamma^C(\rho))]\right)^p\right]^{1/p} \\ \nonumber
&=&\left[\frac{1}{2}\sum_{I}\Tr\left(M_{I}^{\Phi^{(B)}}(\rho){M_{I}^{\Phi^{(B)}}}^\dagger(\rho) \right)^{\frac{p}{2}}\sum_{I^\prime}\left(\rho^C_{i^\prime j^\prime}\rho^C_{j^\prime i^\prime}\rho^C_{k^\prime l^\prime}\rho^C_{l^\prime k^\prime}\right)^{\frac{p}{2}}\right]^{1/p}\\ \label{D-Lambda}
&=&D_{p}^{\Phi^{(B)}}(\rho)\Lambda^{\Phi^{(C)}}_p(\rho^C).
\end{eqnarray}
Here $\sum_I$ is sum over all values of the pairs $I=\{ij,kl\}$ and $\sum_{I^\prime}$ is defined similarly. Also  we have defined $\Lambda^{\Phi^{(C)}}_p(\rho^C)=\left(\sum_{i^\prime j^\prime}\left|\rho^C_{i^\prime j^\prime}\right|^{p}\right)^{2/p}$. Note that factorization in Eq. \eqref{D-Lambda} arisen because we used a product basis for $\mathcal{H}^B\otimes \mathcal{H}^C$ which, of course, is not the general one.  However, the product basis is sufficient to get the minimum, because in this case the problem of finding minima is reduced to find minimums of two independent terms.
Therefore
\begin{eqnarray}\nonumber
D_{p}(\rho\otimes \rho^C)&=&\min_{\Phi^{BC}}D_{p}^{\Phi^{(BC)}}(\rho\otimes \rho^C)=
\min_{\Phi^{(B)}}D_{p}^{\Phi^{(B)}}(\rho)\min_{\Phi^{(C)}}\Lambda^{\Phi^{(C)}}_p(\rho^C) \\
&=&D_{p}(\rho)\Lambda_p(\rho^C),
\end{eqnarray}
where $D_{p}(\rho)$ is given by Eq. \eqref{DpRho} and we have defined $\Lambda_p(\rho^C)=\min_{\Phi^{(C)}}\Lambda^{\Phi^{(C)}}_p(\rho^C)$. Now to make any progress we have to find $\Lambda_p(\rho^C)$.  Evidently, for $p=2$ we get $\Lambda^{\Phi^{(C)}}_{p=2}(\rho^C)=\mu(\rho^C)\le 1$ with $\mu(\rho^C)=\Tr(\rho^C)^2$ as the purity of $\rho^C$, which is independent of the basis $\Phi^{(C)}$. Moreover for $p\le 2$, the minimum of $\Lambda^{\Phi^{(C)}}_p(\rho^C)$ is achieved when $\Phi^{(C)}$ coincides with the eigenvectors of $\rho^C$, so that we get
\begin{eqnarray}\nonumber
\Lambda_p(\rho^C)&=&\min_{\Phi^{(C)}}\Lambda^{\Phi^{(C)}}_p(\rho^C)=\min_{\Phi^{(C)}}\left(\sum_{i^\prime j^\prime}\left|\rho^C_{i^\prime j^\prime}\right|^{p}\right)^{2/p}=\left(\sum_{i^\prime}(\lambda^C_{i^\prime})^{p}\right)^{2/p},\quad \text{for}\quad p\le 2,
\end{eqnarray}
where $\lambda^C_{i^\prime}$ denotes eigenvalues of $\rho^C$. In particular for $p=1$ we find  $\Lambda_{p=1}(\rho^C)=\Tr{\rho^C}=1$, i.e. $D_{p=1}(\rho)$ is invariant under local reversible operations on the unmeasured subsystem.

\end{enumerate}
As it is evident from the last line of the proof of the property 3, the required property of being invariant under local unitary transformations is satisfied by $D_{p}(\rho)$, but the quantity $D_{p}^{\Phi^{(B)}}(\rho)$ lacks this essential property. This means that $D_{p}^{\Phi^{(B)}}(\rho)$ depends, in general, on the orthonormal basis of the $\mathcal{H}^B$, so that $D_{p}^{\Phi^{(B)}}(\rho)$ can not be considered as a \textit{bona fide} measure of quantum correlation. Moreover, at least for $d\le 3$ and $p=1$, such defined measure of quantum correlation  leads to an entanglement monotone when we consider pure states.
 This, therefore, indicates  that the measures considered in Refs. \cite{Wu2011, Guo2012} have the drawback that  they are either   base-dependent, i.e. it is not invariant under local unitary transformations performed on the subsystems \cite{Wu2011}, or does not reduce to an entanglement monotone for pure states \cite{Guo2012}. Moreover, they may increase under reversible actions performed on the subsystem $B$ whose classicality is not tested \cite{Piani2012}.

\section{Examples}
In this section we provide two illustrative examples.

{\it  Two-qubit Werner states.---} As the first example let us consider the two-qubit  Werner state
  \begin{equation}\label{Werner state}
\rho =\frac{(2-a)}{6}I+(\frac{2a-1}{6})\textbf{F},\quad a\in [-1,1],
\end{equation}
with $ \textbf{F}=\sum_{k,l=0}^{1}(\ket{kl}\bra{lk}) $. In this case we obtain
  \begin{equation}\label{our Werner state}
D_{p=1}(\rho)=\frac{1}{6}(1-2a)^{2}=\sqrt{6}D_{p=2}(\rho) =3D_G(\rho)=3D^1_G(\rho),
\end{equation}
where $D_G(\rho)$ is the original geometric discord \cite{Dakic2010}, and $D^1_G(\rho)$ is the one-norm geometric discord \cite{PaulaPRA2013}.

{\it Two-qubit quantum-classical states.---}
Let us now consider the quantum-classical states defined by
 \begin{equation}\label{TDD density mat }
\rho=p \rho_{0}^A\otimes\ket{0}\bra{0}+(1-p) \rho_{1}^A\otimes\ket{1}\bra{1},\quad 0\le p \le 1,
\end{equation}
where $\rho_{i}^A=\frac{1}{2}(\Id+\vec{s_{i}}\cdot\vec{\sigma})$ for $i=0,1$ with $\vec{s}_{0}=(0,0,s_{0})$,  $\vec{s}_{1}=(s_{1}\sin\varphi,0,s_{1}\cos\varphi)$, and  $0\le \varphi\le \pi$ \cite{Ciccarello2014}.  The one-norm geometric discord of these states is given by \cite{Ciccarello2014}
\begin{equation}\label{TDD tc paper }
D^{1}_{G}(\rho) =\frac{\sin\varphi}{2}\min\lbrace ps_{0},(1-p)s_{1}\rbrace.
\end{equation}
For these states we have
 \begin{equation}\label{TDD D2 }
D_{p=2}(\rho) =\frac{p(1-p) s_{0}s_{1}\sin{\varphi}}{\sqrt{2}}=\frac{Q(\rho)}{4\sqrt{2}},
\end{equation}
where $Q(\rho)$ is a measure of quantum correlation of quantum-classical states defined by Abad \etal \cite{Abad2012} as
 \begin{equation}\label{TDD karimpour }
Q(\rho) =4p(1-p)\vert\vec{s_{0}}\times\vec{s_{1}}\vert.
\end{equation}
For the purpose of calculation of $D_{p=1}^{\Phi^{(B)}}(\rho)$, let us choose  $\Phi^{(B)}=\{\ket{\phi_i^B}\}_{i=1}^{2}$, where $\ket{\phi_1^B}=\cos \theta \ket{0}+e^{i\phi}\sin \theta \ket{1}$ and $\ket{\phi_2^B}=\sin \theta \ket{0}-e^{i\phi}\cos \theta \ket{1}$, as a general basis for $\mathcal{H}^B$. In this manner we obtain
\begin{equation}\label{D1QCs }
D_{p=1}^{\Phi^{(B)}}(\rho) =\left\{\vert\cos2\theta\vert+2\vert\sin2\theta\vert\right\}p(1-p) s_{0}s_{1}\sin\varphi,
\end{equation}
which, clearly, depends on the chosen basis $\Phi^{(B)}$ via $\theta$.  Using the fact that minimum of $\left\{\vert\cos2\theta\vert+2\vert\sin2\theta\vert\right\}$ occurs at $\theta=0$ or $\pi/2$, we find
\begin{equation}
D_{p=1}(\rho) =p(1-p) s_{0}s_{1}\sin\varphi=\sqrt{2} D_{p=2}(\rho) =\frac{Q(\rho)}{4}.
\end{equation}
Furthermore, geometric discord of this state can be written as
\begin{equation}
D_G(\rho)=\frac{1}{4} \left(p^2 s_0^2+(1-p)^2 s_1^2-\sqrt{p^4 s_0^4+2 p^2(1-p)^2 s_0^2 s_1^2 \cos{2 \varphi}+(1-p)^4 s_1^4}\right).
\end{equation}
Figure \ref{QCstate} clarifies the comparison between these measures.

\begin{figure}[ht!]
\centering
\includegraphics[width=18cm]{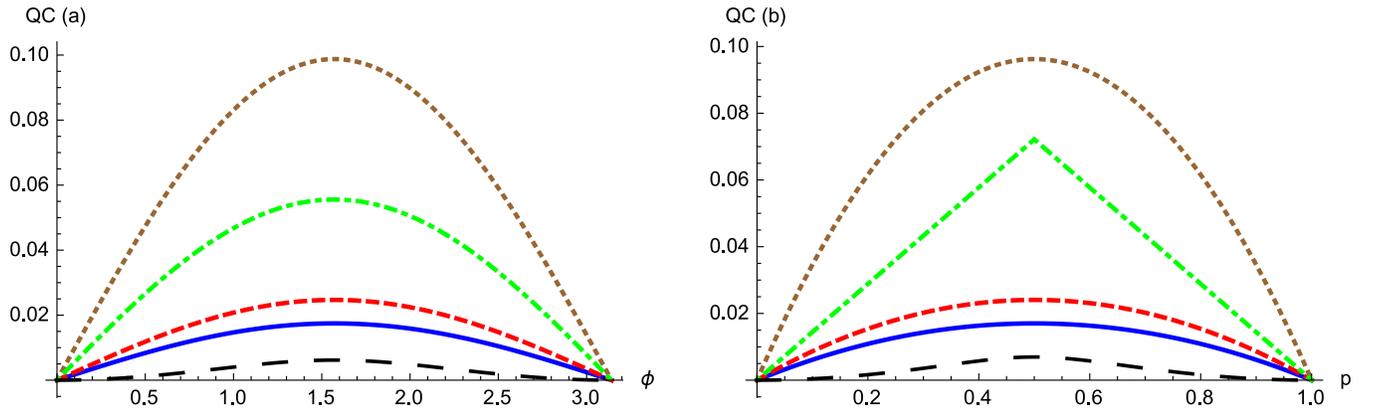}
\caption{(Color online) Quantum correlations $D_{p=2}(\rho) $ (blue, solid line), $D_{p=1}(\rho)$ (red, dashed line), $ Q(\rho) $ (brown, dotted line), $D^{1}_G(\rho)$ (green, dashed-doted line) and $D_G(\rho)$ (black, long-dashed line) of the  quantum-classical state \eqref{TDD density mat } for $s_{0}=s_{1}=\frac{1}{3}$, and (a) $p=\frac{2}{3}$,  (b)  $\varphi=\frac{\pi}{3}$.
}
\label{QCstate}
\end{figure}

\section{Conclusion}
In summary, we have defined a class of quantum correlation identifiers with the aid of  the concept of the nondisruptive local state identification. This concept provides the necessary and sufficient condition for classicality of correlation for bipartite quantum systems. Any departure from this condition could be considered as a measure for quantumness of correlation between two parts of such systems. We have employed  the general Schatten $p$-norm to evaluate how the above condition is violated. Moreover, we have looked  at the  measures from the information-theoretic point of view. For this purpose, we have checked some  properties such as: invariance under local unitary transformation,  monotonicity under local quantum operation and classical communication, and invariance under a local and reversible operations performed on the unmeasured subsystem. For the case of Hilbert-Schmidt norm,  i.e. $p=2$, we have obtained, without any optimization procedure,   an explicit closed formula for measuring the quantum correlation of an arbitrary bipartite state. For the other cases ($p\neq 2$) the optimization procedure is needed in general,  reducing therefore the computability of the measures. However, for general pure states, we have shown  that our measures are bounded form above. Furthermore, it is shown that this upper bound is an entanglement monotone so that, by using the method of convex roof construction, serve a new measure of entanglement for a  general bipartite state. In order to clarify and compare our measures with the other measures, we have provided two two-qubit examples.

\acknowledgments
The authors wish to thank The Office of Graduate
Studies of The University of Isfahan for their support.

\appendix
\section{Proof of theorem \ref{D2FATheorem}}\label{AppendixProofD=2}
In this appendix we provide a proof for Theorem \ref{D2FATheorem}.
Let $\{\hat{\lambda}_i^A\}_{i=1}^{{d_A}^2-1}$ and  $\{\hat{\lambda}_j^B\}_{j=1}^{{d_B}^2-1}$ be generators of $SU(d_A)$ and $SU(d_B)$, respectively, fulfilling the following relations
\begin{eqnarray}\label{SUmGellMann}
 \Tr{\hat{\lambda}_i^s}=0,\qquad \Tr(\hat{\lambda}_i^s\hat{\lambda}_j^s)=2\delta_{ij},\qquad [\hat{\lambda}_i^s,\hat{\lambda}_j^s]=i\sum_{k=1}^{d_s^2-1}f^s_{ijk}\hat{\lambda}_k^s, \qquad s=A,B,
\end{eqnarray}
where we have defined $f^s_{ijk}$ as the structure constant of the Lie algebra $su(d_s)$ \cite{GeorgiBook1999}.
Then a general bipartite state $\rho$ on $\mathcal{H}^{A}\otimes \mathcal{H}^{B}$
can be written in this basis as
\begin{eqnarray}\label{Rho-Bipart2}
\rho=\frac{1}{d_Ad_B}\left(\Id^A\otimes \Id^B +\vec{x}\cdot\hat{\lambda}^A\otimes \Id^B+\Id^A\otimes {\vec y}\cdot\hat{\lambda}^B
+\sum_{i=1}^{{d_A}^2-1}\sum_{j=1}^{{d_B}^2-1}t_{ij} \hat{\lambda}_{i}^A\otimes \hat{\lambda}_{j}^B\right).
\end{eqnarray}
Here $\Id^s$ stands for the unit matrix of the Hilbert space $\mathcal{H}^s$, $\vec{x}=(x_1,\cdots,x_{{d_A}^2-1})^{\T}$ and $\vec{y}=(y_1,\cdots,y_{{d_B}^2-1})^{\T}$ are local coherence vectors of the subsystems $A$ and $B$, respectively
\begin{eqnarray}\label{xy}
x_i &=&\frac{d_A}{2}\Tr{\left[(\hat{\lambda}_{i}^A\otimes \Id^B)\rho\right]},\quad
y_j =\frac{d_B}{2}\Tr{\left[(\Id^A\otimes \hat{\lambda}_{j}^B)\rho\right]},
\end{eqnarray}
and $T=(t_{ij})$ is the correlation matrix
\begin{eqnarray}\label{T}
t_{ij}=\frac{d_A d_B}{4}\Tr{\left[(\hat{\lambda}_{i}^A\otimes \hat{\lambda}_{j}^B)\rho\right]}.
\end{eqnarray}
Now, let $\{\ket{j}\}_{j=0}^{d_B-1}$ be the standard canonical basis of the subsystem $B$. Then a general basis of $\mathcal{H}^{B}$ can be written as
$\ket{\phi_i^{(B)}}=U_B\ket{i}$ where $U_B\in{SU(d_B)}$ is a unitary matrix acting on $\mathcal{H}^{B}$.
In this basis and using Eq. \eqref{Aij} we find
\begin{eqnarray} \label{comp general 2 qubit state}
A_{ij}^{\Phi{^{(B)}}}(\rho)&=&\bra{i}U_B^{\dag}\rho U_B\ket{j}\\ \nonumber&=&
\frac{1}{d_A d_B}\left\{\left(\delta_{ij}+\sum_{l=1}^{{d_B}^2-1}y_l(U_B^\dagger\hat{\lambda}_l^B U_B)_{ij}\right)\Id^A
+\sum_{k=1}^{d_A^2-1}\left(x_k\delta_{ij}+\sum_{l=1}^{d_B^2-1}t_{kl} (U_B^\dagger\hat{\lambda}_{l}^B U_B)_{ij}\right)\hat{\lambda}_k^A\right\}.
\end{eqnarray}
Using this and after some tedious but straightforward calculations  we find that
\begin{eqnarray}\label{D2FA}
[D_{p=2}^{\Phi^{(B)}}(\rho)]^2=\frac{1}{2}\sum_{I}\Tr\left(M_{I}^{\Phi^{(B)}}(\rho){M_{I}^{\Phi^{(B)}}}^{\dag}(\rho)\right)
=-\frac{4}{d_A^4 d_B^4}\Tr{\left\{\mathcal{F}^A(\rho)\left[d_B \vec{x}\vec{x}^\T+TT^\T\right]\right\}},
\end{eqnarray}
where we have defined
\begin{equation}\label{FA}
\mathcal{F}^A(\rho)=\sum_{r=1}^{d_A^2-1}\left[F^A_r(TT^\T){F^A_r}^\dagger\right].
\end{equation}
Here $\{F^A_r\}_{r=1}^{d_A^2-1}$ is the adjoint representation of the $su(d_A)$ Lie algebra  defined as $(F^A_r)_{pq}=-if^A_{pqr}$ with $f^A_{pqr}$ as the structure constants of the algebra given by Eq. \eqref{SUmGellMann}.
It may be useful to obtain the matrix $\mathcal{F}^A(\rho)$ for $d_A=2$. In this case we find
\begin{equation}
\mathcal{F}^A(\rho)=4\left(\begin{array}{ccc}-(TT^\T)_{22}-(TT^\T)_{33} & (TT^\T)_{12} & (TT^\T)_{13} \\
(TT^\T)_{12} & -(TT^\T)_{11}-(TT^\T)_{33} & (TT^\T)_{23} \\
(TT^\T)_{13} & (TT^\T)_{23} & -(TT^\T)_{11}-(TT^\T)_{22}
\end{array}\right),
\end{equation}
where $(TT^\T)_{ij}$ denotes the matrix elements of $TT^\T$. It turns out from Eqs. \eqref{D2FA} and \eqref{FA} that $D_{p=2}^{\Phi^{(B)}}(\rho)$ does  not depend on the chosen basis $\Phi^{(B)}$, so that we arrive at the theorem \ref{D2FATheorem}.

\end{document}